\documentclass[11pt]{amsart}
\usepackage{amsmath,amsthm,amssymb,amsfonts,verbatim}
\usepackage[margin =1.1in]{geometry}

\usepackage{hyperref}

 \providecommand{\F}{\mathbb{F}}

\title{Constructions of maximally recoverable local reconstruction  codes  via function fields}

\author{Venkatesan Guruswami \and Lingfei Jin \and Chaoping Xing}

\thanks{V.~G. is with the Computer Science Department, Carnegie Mellon University, Pittsburgh, PA 15213, USA. {\it Email:} {guruswami@cmu.edu}. Some of this work was done when the author was visiting the School of Physical \&  Mathematical Sciences, Nanyang Technological University, Singapore.  Research supported in part by NSF grants CCF-1422045 and CCF-1563742.}
\thanks{L.~J. is with the School of Computer Science, Fudan University, Shanghai, China. {\it Email:} {lfjin@fudan.edu.cn}.}
\thanks{C.~X. is with the Division of Mathematical Sciences, School of Physical \&  Mathematical Sciences, Nanyang Technological University, Singapore. {\it Email:} {xingcp@ntu.edu.sg}.}

\newtheorem{lemma}{Lemma}[section]
\newtheorem{theorem}[lemma]{Theorem}

\newtheorem{defn}{Definition}

\theoremstyle{remark}
\newtheorem{rmk}{Remark}

\newcommand{\eps}{\varepsilon}
\renewcommand{\epsilon}{\varepsilon}
\renewcommand{\le}{\leqslant}
\renewcommand{\ge}{\geqslant}

\def\ZZ{\mathbb{Z}}

\def\PP{\mathbb{P}}

\def \mL {\mathcal{L}}

\def\Pin{{P_{\infty}}}

\newcommand{\Ga}{\alpha}
\newcommand{\Gb}{\beta}

\newcommand{\Ge}{\epsilon}

\newcommand{\Gl}{\lambda}

\def \ba {{\bf a}}

\def \bh {{\bf h}}
\def \bo {{\bf 0}}

\def\g{{\mathfrak{g}}}

\begin{document}

\maketitle
\thispagestyle{empty}

\begin{abstract}
	
	Local Reconstruction Codes (LRCs) allow for recovery from a small number of erasures in a local manner based on just a few other codeword symbols. They have emerged as the codes of choice for large scale distributed storage systems due to the very efficient repair of failed storage nodes in the typical scenario of a single or few nodes failing, while also offering fault tolerance against worst-case scenarios with more erasures.
	A maximally recoverable (MR) LRC offers the best possible blend of such local and global fault tolerance, guaranteeing recovery from all erasure patterns which are information-theoretically correctable given the presence of local recovery groups. In an  $(n,r,h,a)$-LRC, the $n$ codeword symbols are partitioned into $r$ disjoint groups each of which include $a$ local parity checks capable of locally correcting $a$ erasures. The codeword symbols further obey $h$ heavy (global) parity checks. Such a code is maximally recoverable if it can correct all patterns of $a$ erasures per local group plus up to $h$ additional erasures anywhere in the codeword.  This property amounts to linear independence of all such subsets of columns of the parity check matrix.
	
	\smallskip
	MR LRCs have received much attention recently, with many explicit constructions covering different regimes of parameters. Unfortunately, all known constructions require a large field size that exponential in $h$ or $a$, and it is of interest to obtain MR LRCs of minimal possible field size. In this work, we develop an approach based on function fields to construct MR LRCs. Our method recovers, and in most parameter regimes improves, the field size of previous approaches. For instance, for the case of small $r \ll \epsilon \log n$ and large $h \ge \Omega(n^{1-\epsilon})$, we improve the field size from roughly $n^h$ to $n^{\epsilon h}$. For the case of $a=1$ (one local parity check), we improve the field size quadratically from $r^{h(h+1)}$ to $r^{h \lfloor (h+1)/2 \rfloor}$ for some range of $r$. The improvements are modest, but more importantly are obtained in a unified manner via a promising new idea. We associate distinct places with different local groups, and for each group, use functions with a single pole at that place and define a Moore matrix as the portion of the parity check matrix corresponding to that group. The requisite linear independence is established using properties of the Moore determinant to reduce linear independence over an extension field to that over the base field. The latter is established using the distinctness of places across groups, and by direct code based design within a group.

	
\end{abstract}


\newpage
\section{Introduction}
Interest in erasure codes has surged in recent years, with the demands of massive cloud storage systems raising hitherto unexplored, yet very natural and mathematically deep, questions concerning the parameters, robustness, and efficiency of the code. Distributed storage systems need to build in redundancy in the data stored in order to cope with the loss or inaccessibly of the data on one or more storage nodes. Traditional erasure codes offer a natural strategy for such robust data storage, with each storage node storing a small part of the codeword, so that the data is protected against multiple node failures.
In particular, MDS codes such as Reed-Solomon codes can operate at the optimal storage vs. reliability trade-off --- for a given amount of information to be stored and available storage space, these codes can tolerate the maximum number of erasures without losing the stored information.

Individual storage nodes in a large scale system often fail or become unresponsive. Reconstruction (repair) of the content stored on a failed node with the help of remaining active nodes is important to reinstate the system in the event of a permanent node failure, and to allow access to the data stored on a temporarily unavailable node. The use of erasure codes in large storage systems, therefore, brings to the fore a new requirement: the ability to \emph{very efficiently} reconstruct \emph{parts} of a codeword from the rest of the codeword.

Local Reconstruction Codes (LRCs), introduced in \cite{GHJY14}, offer an attractive way to meet this requirement. An LRC imposes local redundancies in the codewords, so that a single (or a small number of) erased symbol can be recovered locally from less than $r$ other
codeword symbols.\footnote{LRCs are also expanded as Locally Repairable Codes or Locally Reoverable Codes, eg. \cite{PD14,TB,GXY18}.} Here $r$ is the locality parameter that is typically much smaller than the code length $n$. In the distributed storage context, an LRC allows for the low-latency repair of any failed node as one only needs to wait for the response from $r$ nodes.  LRCs have found
spectacular practical applications with their use in the Windows
Azure storage system~\cite{HSX12}.

The challenge in an LRC design is to balance the locality requirement, that allows fast recovery from a single or few erasures, with good global erasure-resilience (via traditional slower methods) for more worst-case scenarios. One simple metric for global fault tolerance is the minimum distance $d$ of the code, which means that any pattern of fewer than $d$ erasures can be corrected. The optimal trade-off between the distance, redundancy, and locality of an LRC was established in \cite{GHSY12}, and an elegant sub-code of Reed-Solomon codes meeting this bound was constructed in \cite{TB}.

This work concerns a much stronger requirement on global fault-tolerance, called {\em Maximal Recoverability}. This requires that the code should simultaneously correct every erasure pattern that is information-theoretically possible to correct, given the locality conditions imposed on the codeword symbols. Let us describe it more formally in the setting of interest in this paper. Define an $(n,r,h,a)_\ell$-LRC to be a linear code over $\F_\ell$ of length $n$
whose $n$ codeword symbols are partitioned into $r$ disjoint groups each of which include $a$ local parity checks capable of locally correcting $a$ erasures. The codeword symbols further obey $h$ heavy (global) parity checks. With this structure of parity checks, it is not hard to see that the erasure patterns one can hope to correct are precisely those which consist of up to $a$ erasures per local group plus up to $h$ additional erasures anywhere in the codeword. A \emph{maximal recoverable} (MR) LRC is a \emph{single} code that is capable of simultaneously correcting \emph{all} such patterns. Thus, an MR code gives the most bang-for-the-buck for the price one pays for locality.

This notion was introduced in \cite{BHH13} motivated by applications to storage on solid-state devices, where it was called partial MDS codes. The terminology maximally recoverable codes was coined in \cite{GHJY14}, and the concept was more systematically studied in \cite{GHJY14,GHK17}.  By picking the coefficients of the heavy parity checks randomly, it is not hard to show the existence of MR LRCs over \emph{very large} fields, of size exponential in $h$.  An explicit construction over such large fields was also given in \cite{GHJY14}, which also proved that random codes \emph{need} such large field sizes with high probability.\footnote{This is akin to what happens for random codes to have the MDS property.  However, for MDS codes, the Vandermonde construction achieves a linear field size explicitly.}

Since encoding a linear code and decoding it from erasures involve performing numerous finite field arithmetic operations, it is highly  desirable to have codes over small fields (preferably of characteristic 2). Obtaining MR LRCs over finite fields of minimal size has therefore emerged as a central problems in the area of codes for distributed storage. So far, no construction of MR LRCs that avoids the exponential dependence on $h$ has been found. A recent lower bound shows that, unlike MDS codes, for certain parameter settings one cannot have MR LRCs over fields of linear size. This shows that the notion of maximal recoverability is quite subtle, and pinning down the optimal field size is likely a deep question. There remains a large gap between the upper and lower bounds on field size of MR LRCs, closing which is a challenge of theoretical and practical importance.

In this work, we develop a novel approach to construct MR LRCs based on function fields. Our framework recovers and in fact slightly improves most of the previous bounds in the literature in a unified way.  We note that since there are at least three quantities of significance --- the locality $r$, the local (intra group) erasure tolerance $a$, and number of global parity checks $h$ --- the landscape of parameters and different constructions in this area is quite complex. Also, depending on the motivation, the range of values of interest of these parameters might be different. For example, if extreme efficiency of local repair is important, $r$ should be small. But on the other hand this increase the redundancy and thus storage requirement of the code, so from this perspective a modest $r$ (say $\sqrt{n}$) might be relevant. If good global fault tolerance is required, we want larger $h$, but then the constructions have large field size. It is therefore of interest to study the problem treating these as independent parameters, without assumptions on their relative size. We next review the field size of previous constructions, and then turn to the parameters we achieve in different regimes.

\subsection{Known field size bounds}
For $a\in\{0,r-1\}$,  optimal maximally recoverable local reconstruction codes (MR LRCs, for short) can be constructed  by using either Reed-Solomon codes or their repetition. For $h\le 1$, constructions of maximally recoverable LRCs over fields of size $O(r)$ were given in \cite{BHH13}. For the remaining case: $1\le a\le r-2$ and $h\ge 2$, there are quite number of constructions in literature \cite{BHH13,Bla13,TPD16,GHJY14,HY16,GHK17,CK17,BPSY16,GYBS17,GGY17}.

For the cases of $h=2$ and $h=3$, the best known constructions of MR LRCs were given in \cite{GGY17} with field sizes of $O(n)$  and $O(n^3)$ respectively, uniformly for all $r,a$. (Their field sizes were worse by $n^{o(1)}$ factors compared to these bounds when the field is required to be of characteristic $2$.) For most other parameter settings, the best constructions by \cite{GYBS17} provide a family of MR LRCs over fields of sizes
\begin{equation}\label{eq:1}
\ell=O\left(r\cdot n^{(a+1)h-1}\right)
\end{equation}
as well as
\begin{equation}\label{eq:2}
\ell= \max\left\{O(\frac nr),\ O(r)^{h+a}\right\}^h \ ,
\end{equation}

The bound \eqref{eq:1} outperforms the bound \eqref{eq:2} when $r=\Omega(n)$, while  the bound \eqref{eq:2} is better when $r\ll n$. In both the bounds, the field size grows exponentially with $h$ and $a$.

Recently, by using maximum rank distance (MRD) codes, the paper \cite{NH18} (specifically Corollary 14) gives a family of MR LRCs over fields of sizes
\begin{equation}\label{eq:3}
\ell=O\left(r^{\frac{n(r-a)}r}\right).
\end{equation}
When $r=\Omega(n)$, and $a$ is close to $r$ or $h$ is large, \eqref{eq:3} is better than bounds \eqref{eq:1} or \eqref{eq:2}. By using probabilistic arguments,  the paper \cite{NH18} shows existence of a family of MR LRCs over fields of sizes
\begin{equation}\label{eq:3a}
\ell=O\left({n-1\choose k-1}\right),
\end{equation}
where $k=n\left(1-\frac ar\right)-h$ is the dimension of the code.

On the other hand, a lower bound on the field size was presented in \cite{GGY17}. Stating the bound when $h \le \frac{n}{r}$ for simplicity, they show that the field size $\ell$ of an $(n,r,h,a)_\ell$ MR LRC must obey
\begin{equation}\label{eq:4}
\ell=\Omega_{a,h}\left(n\cdot r^{\min\{a,h-2\}}\right).
\end{equation}
 The lower bound \eqref{eq:4} is still quite far from the upper bounds \eqref{eq:1} and \eqref{eq:2}. In particular, the exponent of $a$ or $h$ is to the base growing with $n$ in the known constructions, but only to the base $r$ in the above lower bound. Thus, one can conjecture that there is still room to improve both the constructions and the lower bounds. We note that under more complex structural requirements on the local groups, notably grid-like topologies and product codes, the optimal field size has been pinned down to $\exp(\Theta(n))$~\cite{KLR-grid-MR}.

 Several techniques have been employed in literature for constructions of MR LRCs. One prevalent idea is to use a ``linearized" version of the Vandermonde matrix, where the heavy parity check part of the matrix consists of columns $(\alpha_i, \alpha_i^q, \dots, \alpha_i^{q^{h-1}})^T$ where $\alpha_i \in \F_\ell$ for a sufficiently high degree extension field $\F_\ell$ of $\F_q$.
 This construction is combined with $2h$-wise independent spaces to get an $O(n^h)$ field size in \cite{GHJY14}, and is also employed in \cite{GYBS17}.  Another approach is based on rank-metric codes (see, for instance, \cite{CK17,NH18}). Various ad hoc methods have been employed for good constructions of MR LRCs for small $h$, for example for $h=2,3$ in \cite{GGY17}.

\subsection{Our results} In this work, we develop a new approach to construct MR LRCs based on algebraic function fields. We discuss the key elements underlying our strategy in Section~\ref{subsec:techniques}, but for now state the field sizes of the MR LRCS we can construct  for various regimes of parameters.  Most of the existing results in literature can be recovered through our methods in a unified way. In most regimes, the parameters of our codes beat the known ones. For easy reference, we summarize the different possible trade-offs we can achieve in one giant theorem statement below. Since this comprehensive statement may be overwhelming to parse, let us highlight just two of our significant improvements: item (i) for $a=1$, where we improve $r^{h+1}$ term in \eqref{eq:2} quadratically to $r^{\lfloor \frac{h+1}{2}\rfloor}$, and item (vi) for sufficiently large $h$, where the exponent $h$ in bounds \eqref{eq:1} and \eqref{eq:2} is improved to $\epsilon h$. Also the exponent $h$ is replaced by $\min\{h,n/r\}$ in the bounds (i)-(iv) that improve \eqref{eq:2}. In the bounds (vii) and (viii) the factor $n/r$ in the exponent is improved to $\min\{k,n/r\}$; this improved is less significant as it only applies to the low-rate setting but included for completeness and also to reflect a construction approach based on generator matrices (as opposed to parity check matrices which is a more potent way to reason about MR LRCs that underlies the other parts of the theorem).


\begin{theorem}\label{thm:1.1} One has a maximally recoverable $(n,r,h,a)_\ell$-local reconstruction code over a field of size $\ell$ with parameters satisfying any of the following conditions. (Below $\tilde{O}(f)$ denotes $f \log^{O(1)} f$.)
\begin{itemize}
\item[{\rm (i)}] {\rm (see Theorem \ref{thm:3.9})} $a=1$,$r\ge h+2$ and
\[\ell \le \left(\max\left\{\tilde{O}(\frac n{r}), (2r)^{\left\lfloor\frac{h+1}2\right\rfloor}\right\}\right)^{\min\{h,\frac nr\}} \mbox{ and $\ell$ is even};\]
\item[{\rm (ii)}] {\rm (see Theorem \ref{thm:3.10})} $a=1$ and
\[\ell \le \left( \max\left\{ \tilde{O}(\frac nr), 2^r\right\}\right)^{\min\{h,\frac nr\}} \mbox{ and $\ell$ is even};\]
\item[{\rm (iii)}] {\rm (see Theorem \ref{thm:3.12})} for all settings of $n,r,h,a$ and
\[\ell \le \left(\max\left\{\tilde{O}(\frac nr), (2r)^{h+a}\right\}\right)^{\min\{h,\frac nr\}} ;\]
\item[{\rm (iv)}] {\rm (see Theorem \ref{thm:3.13})} for all settings of $n,r,h,a$ and
\[ \ell \le \left(\max\left\{\tilde{O}(\frac nr), (2r)^r\right\}\right)^{\min\{h,\frac nr\}};\]
\item[{\rm (v)}]{\rm (see Theorem \ref{thm:4.3})} $r = O\left(\frac{\log n}{\log\log n}\right)$ and $hr \ge \Omega\left( \frac{ n^{\frac23}}{\Ge}\right)$ for a positive real $\Ge\in(0,0.5)$ and \[\ell \le O\left(n^{\frac{2h}{3}\left(1+{\Ge}\right)}\right);\]
\item[{\rm (vi)}] {\rm (see Theorem \ref{thm:4.4})} $r= O\left(\frac{\Ge\log n}{\log\log n}\right)$ and $hr=\Omega\left(n^{1-\Ge}\right)$ for a positive real $\Ge\in(0,0.5)$ and \[\ell \le n^{\Ge h};\]
\item[{\rm (vii)}] {\rm (see Theorem \ref{thm:3.3})} for all settings of $n,r,h,a$
 \[
\ell \le \left\{\begin{array}{ll}
2^{\min\left\{rk,n \right\}} \le 2^{ n} &\mbox{ if $r\ge \log n$} \\
2^{\lceil\log n\rceil  \min \{k, \frac nr \}}&\mbox{ if $r\le \log n$}
\end{array}\right. \]
 where $k=\left(1-\frac ar\right)-h$ is the dimension of the code;
 \item[{\rm (viii)}] {\rm (see Theorem \ref{thm:3.5})} $r-a=\Omega(\log n)$ and
 \[\ell
 \le 2r^{\lfloor\frac{r-a}2\rfloor \min\{k,\frac {n}{r}\}}   \mbox{ and $\ell$ is even}.\]
\end{itemize}
\end{theorem}


The first two bounds, and the bounds in (vii) and (viii) of Theorem \ref{thm:1.1} are derived from the rational function fields $\F_2(x)$. In addition, the bounds in (i) and (viii) of Theorem \ref{thm:1.1} are obtained via a combination with binary BCH codes. The bounds in (iii) and (iv) of Theorem \ref{thm:1.1} are derived from rational function field $\F_q(x)$, where $\ell$ is a power of $q$.
The fifth bound is obtained via Hermitian function fields, while the sixth bound is derived from the Garcia-Stichtenoth function field tower.
Our codes achieving the trade-offs stated in the above theorem can in fact be explicitly specified. But we note that for MR codes even existence questions over small fields are interesting and non-trivial.

\subsection{Comparison.} Each of our bounds in Theorem \ref{thm:1.1} beats the known results in some parameter regimes. Let us compare them one by one.

\begin{itemize}
	\item
The bound in Theorem \ref{thm:1.1}(i) outperforms the bound \eqref{eq:2} due to the quadratically better exponent for $r$.



\item
The bound in Theorem \ref{thm:1.1}(ii) outperforms even the bound in Theorem \ref{thm:1.1}(i) for $\frac r{\log r}<\left\lfloor\frac{h+1}2\right\rfloor$.

\item
The bound in Theorem \ref{thm:1.1}(iii) outperforms the bound \eqref{eq:2} for $h>\frac nr$.

\item
The bound in Theorem \ref{thm:1.1}(iv) even outperforms the bound in Theorem \ref{thm:1.1}(iv) for $r<h+a$, and hence it beats the bound \eqref{eq:2} for $\frac nh<r<h+a$.

\item
The bound in Theorem \ref{thm:1.1}(v) outperforms both the bounds \eqref{eq:1} and \eqref{eq:2}  for all parameter settings subject to $r=\widetilde{O}(\log n)$ and $hr=\Omega\left( \frac{ n^{\frac23}}{\Ge}\right)$. It is clear that  the bound in Theorem \ref{thm:1.1}(v) is better than \eqref{eq:1}. As $r=\widetilde{O}(\log n)$, then we have $\left(\frac nr\right)^h>n^{h(1-o(1))}>n^{2h(1+\Ge)/3}$ and hence the bound in Theorem \ref{thm:1.1}(v) beats \eqref{eq:2} in this case.

\item
As the bound in Theorem \ref{thm:1.1}(vi) is even better than the bound in Theorem \ref{thm:1.1}(v), the bound in Theorem \ref{thm:1.1}(vi) beats both the bounds \eqref{eq:1} and \eqref{eq:2}  for all parameter settings subject to $r=\widetilde{O}(\Ge\log n)$ and $hr=O\left(n^{1-\Ge}\right)$ for a positive real $\Ge\in(0,0.5)$.

\item  When the dimension $k$ is much smaller than $n$, then the probabilistic bound \eqref{eq:3a} gives the field size $O(n^k)=O(2^{k\log n})$ which is the same size as in Theorem \ref{thm:1.1}(vii) for $r\le \log n$.  When the dimension $k$ is proportional to $n$, then  the probabilistic bound \eqref{eq:3a} gives the field size $2^{O(n)}$ which is the same as the bound $2^n$ in Theorem \ref{thm:1.1}(vii) for $r\ge \log n$.
\item Finally, the bound in Theorem \ref{thm:1.1}(viii) clearly outperforms the bound \eqref{eq:3} when $k < n/r$.

\end{itemize}

\subsection{Our techniques}
\label{subsec:techniques}
Note that construction of MR LRCs is equivalent to construction of certain generator or parity-check matrices with requirement of column linear independence (see Section \ref{subsec:2.1}).

Our construction idea departs from previous approaches and is based on function fields over a finite field $\F_q$. The key in constructing an MR LRC is the choice of the heavy parity checks. We now briefly describe our idea to pick these. We associate with each of the $g=n/r$ local groups a distinguishing (high degree) place $P_i$, $1 \le i \le g$. The degree of the place is chosen large enough to guarantee the existence of at least $g$ such places. For each local group, we pick functions $f_{ij}$, $1 \le j \le r$, that have \emph{exactly one pole at $P_i$}.
The coefficients of the $h$ heavy parity checks corresponding to the $j$'th symbol of $i$'th local group are chosen to be
\begin{equation}
\label{eq:pc-col}
(f_{ij}(Q), f_{ij}^q(Q),\dots,f_{ij}^{q^{h-1}}(Q))^T  \ ,
\end{equation}
where $Q$ is a place of sufficiently high degree, so that the evaluations $f_{ij}(Q)$ belong to an extension field $\F_\ell$ which will be the final alphabet size of the MR LRC. By properties of the Moore determinant (Section~\ref{subsec:2.2}) and the large degree of $Q$, the required linear independence of
columns such as \eqref{eq:pc-col} over $\F_\ell$ reduces to a certain linear independence requirement for the $f_{ij}$'s over $\F_q$. Across different local groups such linear independence follows because a function with one pole at $P_i$ cannot cancel a function with one pole at a different place $P_{i'}$. Within a local group, the required linear independence is ensured by choosing the $f_{ij}$'s within a group so that any $h+a$ of them (which is the maximum number of erasures we can have within a group) are linearly independent over $\F_q$.

We remark that all our various guarantees of Theorem~\ref{thm:1.1} except Parts (v) and (vi) are obtained using just the rational function field, and can be described in elementary language using just polynomials, as we do in Section~\ref{sec:rational-FF}.

\subsection{Organization} The paper is organized as follows. In Section 2, we introduce some preliminaries such as MR LRCs (both the generator and parity check matrix viewpoints) and Moore determinants. In Section 3, we present our constructions of MR LRCs using the rational function field together with a concatenation with classical codes of good rate vs. distance trade-off. We give two constructions, using the generator matrix viewpoint in the first part (yielding Parts (vii) and (viii) of Theorem~\ref{thm:1.1}), and then a parity check based construction in second part which yields Parts (i)-(iv) of Theorem~\ref{thm:1.1}.
This section is elementary and only uses properties of polynomials.
In Section 4, we generalize the construction of MR LRCs  via parity-check matrix given in Section 3 by making use of  arbitrary algebraic function fields. The necessary preliminaries on function fields are deferred to this section as we do not need them in Section 3. We then apply this construction to Hermitian function fields and the Garcia-Stichtenoth tower to
obtain MR LRCs promised in Parts (v) and (vi) of Theorem~\ref{thm:1.1} respectively.

\section{Preliminaries}
\subsection{Maximally recoverable local reconstruction codes}\label{subsec:2.1}
Throughout this paper, $\F_q$ denotes the finite field of $q$ elements for a prime power $q$. We use $\F_q^{k\times n}$ to denote the set of all $k\times n$ matrices over $\F_q$.

Consider a distributed storage system where there are $g$ disjoint locality groups and each group has size $r$ and can locally correct any $a$ erasure errors. In addition, the system can correct any $h$ erasure errors together with any $a$ erasure errors in each group. This requires a class of codes called  {\it maximally  recoverable local reconstruction codes}  or {\it partial MDS codes} for error correction of such a system. The precise definition of MR LRCs is given below.
\begin{defn}\label{def:1}{\rm Let $\ell$ be a prime power and let $a,g,r,h$ be positive integers satisfying  $ga+h<gr$. Put $n=gr$ and $k=n-ga-h$. An $\ell$-ary $[n,k]$-linear code with a generator matrix of the form
\[G=(B_1|B_2|\cdots|B_g)\in\F_{\ell}^{k\times n}\]
is called  a maximally recoverable $(n,r,h,a)_\ell$-local reconstruction code (or
an MR $(n,r,h,a)_\ell$-LRC, for short) if
\begin{itemize}
\item[(i)] each $B_i$ has size $k\times r$;
\item[(ii)] the row span of each $B_i$ is an $[r,r-a,a+1]_\ell$-MDS code for $1\le i\le g$ (note that $B_i$ is not a generator matrix of this MDS code in general);
\item[(iii)] after puncturing $a$ columns from each $B_i$, the remaining matrix of $G$ generates an $[n-ga,k,h+1]_\ell$-MDS code.
\end{itemize}
}\end{defn}
From the definition, an MR $(n,r,h,a)_q$-LRC  can correct $h$ erasure errors at arbitrarily positions together with any $a$ erasure errors  in each of $g$ groups.  The following lemma directly follows from Definition \ref{def:1}.
\begin{lemma}\label{lem:2.1}
A matrix $G=(B_1|B_2|\cdots|B_g)\in\F_\ell^{k\times n}$ is a generator matrix of an MR $(n,r,h,a)_\ell$-LRC if and only if every $k\times k$ submatrix $S$ of $G$ with at most $r-a$ columns per block $B_i$ is invertible.
\end{lemma}
One can have an equivalent definition via parity-check matrix.
\begin{defn}\label{def:2}{\rm Let $\ell$ be a prime power and let $a,g,r,h$ be positive integers satisfying  $ga+h<gr$. Put $n=gr$ and $k=n-ga-h$. An $\ell$-ary $[n,k]$-linear code with a parity-check matrix of the form
\begin{equation}\label{eq:5} H=\left(\begin{array}{c|c|c|c}
A_1&O&\cdots&O\\ \hline
O&A_2&\cdots&O \\ \hline
\vdots&\vdots&\ddots&\vdots \\ \hline
O&O&\cdots&A_g \\ \hline
D_1&D_2&\cdots&D_g
\end{array}
\right)\in\F_{\ell}^{(n-k)\times n}\end{equation}
is called an MR $(n,r,h,a)_\ell$-LRC if
\begin{itemize}
\item[(i)] each $A_i$ has size $a\times r$ and each $D_i$ has size $h\times r$;
\item[(ii)] each $A_i$  generates an $[r,a,r-a+1]_\ell$-MDS code for $1\le i\le g$ (note that the nullspace of $A_i$ is $[r,r-a,a+1]_\ell$ code);
\item[(iii)] every $ag+h$  columns consisting of any $a$ columns in each group and other arbitrary $h$ columns are $\F_\ell$-linearly independent.
\end{itemize}
}\end{defn}
  \begin{rmk}
  \begin{itemize}
\item[(i)] To see  equivalence between Definitions \ref{def:1} and \ref{def:2}, we note that each $A_i$ in Definition \ref{def:2} is actually a parity-check matrix of the code generated by $B_i$ given in Definition \ref{def:1}.
\item[(ii)] In this paper, we will  use both Definitions \ref{def:1} and  \ref{def:2} for constructions of MR LRCs. However, the major results of this paper come from the constructions based one  Definition \ref{def:2}, i.e., via parity-check matrices of the required form in \eqref{eq:5}.
\end{itemize}
  \end{rmk}

  \subsection{Moore determinant}
  \label{subsec:2.2}
  Let $\ell$ be a power of $q$. For elements $\Ga_1,\dots,\Ga_h\in\F_\ell$, the Moore matrix is defined by
  \[M=\left(\begin{array}{cccc}
\Ga_1&\Ga_2&\cdots&\Ga_h\\
\Ga_1^q&\Ga_2^q&\cdots&\Ga_h^q\\
\vdots&\vdots&\ddots&\vdots \\
\Ga_1^{q^{h-1}} &\Ga_2^{q^{h-1}} &\cdots&\Ga_h^{q^{h-1}}
\end{array}
\right)\in\F_{\ell}^{h\times h}.\]
  The determinant $\det(M)$ is given by the following formula
  \[\det(M)=\prod_{(c_1,\dots,c_h)}(c_1\Ga_1+\cdots+c_h\Ga_h),\]
  where $(c_1,\dots,c_h)$ runs through all non-zero direction vectors in $\F_q^\ell$. Thus, $\det(M)\neq 0$ if and  only if  $\Ga_1,\dots,\Ga_h$ are $\F_q$-linearly independent.

\section{Explicit  constructions via rational function fields}
\label{sec:rational-FF}
In this section, we only introduce  constructions of MR LRCs from rational function fields. Our description will be self-contained and elementary in terms of polynomials and we won't be requiring any background on algebraic function fields (we have therefore deferred the background on function fields to Section~4 ahead of our more general construction in the next section).

\subsection{Constructions via generator matrix}
In this subsection, we present constructions of MR LRCs using Definition \ref{def:1}, i.e., via generator matrices of MR LRCs.

Let $N_q(d)$ denote the number of monic irreducible polynomials of degree $d$ over $\F_q$. Then one has $\sum_{d|m}dN_q(d)=q^m$ for any $m\ge 1$ (see \cite[Corollary 3.21 of Chapter 3]{LN03}). This gives $\sum_{d|m}N_q(d)\ge \frac{q^m}m$. For each monic irreducible polynomial $p(x)$ of degree $d$ with $d|m$, we get a polynomial $g(x)^{m/d}$ of degree $m$. Thus, for any $g\le \left\lceil\frac{q^m}m\right\rceil$, there are $g$ polynomials $p_1(x),p_2(x),\dots,p_g(x)$ of degree $m$ such that $\gcd(p_i(x),p_j(x))=1$ for all $1\le i\neq j\le g$

Assume that (i) $m\ge r$; or (ii) $m<r$ and there is a $q$-ary $[r,r-m,\ge r-a+1]$-linear code, i.e. there exists a subset of $\F_q^{m}$ of size $r$ such that any $r-a$ elements in this subset are $\F_q$-linearly independent.

Choose $g\le\left\lceil \frac{q^m}m\right\rceil$ polynomials $p_1(x),p_2(x),\dots,p_g(x)$ of degree $m$ such that $\gcd(p_i(x),p_j(x))=1$ for all $1\le i\neq j\le g$. Then for each $1\le i\le g$, we can form an $\F_q$-vector $V_i:=\left\{\frac{f(x)}{p_i(x)}:\; f(x)\in\F_q[x],\; \deg(f(x))\le m-1\right\}$ of dimension $m$. Under our condition on $m$, one can find $r$ functions $g_{i1}(x),\dots,g_{ir}(x)\in V_i$ such that any $r-a$ polynomials out of $\{g_{i1}(x),\dots,g_{ir}(x)\}$ are $\F_q$-linearly independent. Choose an irreducible polynomial $Q(x)\in\F_q[x]$ such that $Q(x)$ is coprime with every $p_i(x)$ for $1\le i\le g$.
For a function $h(x)\in V_i$, we use $h(Q)$ to denote the residue class of $h(x)$ in the residue class field $\F_q[x]/Q(x)\simeq\F_{q^{\deg(Q)}}$.

\begin{lemma}\label{lem:3.1}
Let $T$ be a subset $\{1,2,\dots,g\}$ with $|T|\le \deg(Q)/m$. If $\sum_{i\in T}g_i(Q)=0$ for some functions $g_i\in V_i$, then $g_i=0$ for all $i\in T$.
\end{lemma}
\begin{proof} Write $g_i=\frac{f_i}{p_i}$ for some polynomials $f_i$ with $\deg(f_i)\le m-1$. The equality $\sum_{i\in T}g_i(Q)=0$ implies that $\sum_{i\in T}f_i(x)\prod_{j\in T\setminus\{i\}}p_j(x)$ is divisible by $Q(x)$. As
$\deg(\sum_{i\in T}f_i(x)\prod_{j\in T\setminus\{i\}}p_j(x))\le m|T|-1$, we must have that $\sum_{i\in T}f_i(x)\prod_{j\in T\setminus\{i\}}p_j(x)$ is the zero polynomial. Suppose that $f_t\neq 0$ for some $t\in T$, then we have
\[\sum_{i\in T\setminus\{t\}}f_i(x)\prod_{j\in T\setminus\{i\}}p_j(x)=-f_t(x)\prod_{j\in T\setminus\{t\}}p_j(x).\]
The left hand side of the above equality is divisible by $p_t(x)$, while the right hand side of the above equality is not divisible by $p_t(x)$. This contradiction completes the proof.
\end{proof}

Let $Q$ be an irreducible polynomial in $\F_q[x]$ of degree
\[\min\{km,gm\}=\min\left\{km,\frac {nm}r\right\}=\min\left\{(n-\frac{an}r-h)m,\frac {nm}r\right\}.\]

Define the $k\times r$ matrix $B_i$ as follows.
\begin{equation}\label{eq:10}B_i=\left(\begin{array}{ccccccc}
g_{i1}(Q)& g_{i2}(Q)& \cdots & g_{ir}(Q) \\
g_{i1}^q(Q)& g_{i2}^q(Q)&\cdots & g_{ir}^q(Q) \\
\vdots&\vdots&\vdots& \vdots\\
g_{i1}^{q^{k-1}}(Q)&g_{i2}^{q^{k-1}}(Q)& \cdots & g_{ir}^{q^{k-1}}(Q)
\end{array}
\right)\in \F_{q^{\deg(Q)}}^{k\times r}.
\end{equation}

\begin{lemma}\label{lem:3.2}
Assume that $m\ge r$ or there is a $q$-ary $[r,r-m,\ge r-a+1]$-linear code. Let $B_i$ be the matrix given in \eqref{eq:10}. Put $\ell=q^{\min\left\{(n-\frac{an}r-h)m,\frac {nm}r\right\}}=q^{\min\left\{km,\frac {nm}r\right\}}$ and $G=(B_1|B_2|\cdots|B_g)\in\F_{\ell}^{k\times n}$. Then the $\ell$-ary code $C$ with the generator matrix $G$  is an MR $(n,r,h,a)_\ell$-LRC.
\end{lemma}
\begin{proof} Let $A$ be a $k\times k$ submatrix  of $G$ with at most $r-a$ columns per block $B_i$. By Lemma~\ref{lem:2.1}, it is sufficient to show that $A$ is invertible. It follows from Subsection~\ref{subsec:2.2} that this is equivalent to showing that the first row of $A$ is $\F_q$-linearly independent.

Let $S_i$ be a subset of $\{(i,1),(i,2),\dots,(i,r)\}$  for $i=1,2,\dots,g$ such that the first row of $A$ is $(g_{ij}(Q))_{j\in S_i,1\le i\le g}$. Then
 $\sum_{i=1}^{g}|S_i|=k$ and $|S_i|\le r-a$. Let $T$ be a subset of $\{1,2,\dots,g\}$ such that $S_i\neq\emptyset$ if and only if $i\in T$. Then  $\sum_{i=1}^{g}|S_i|=\sum_{i\in T}|S_i|=k$ and hence $|T|\le \min\{k,g\}$.
Let $\Gl_{ij}\in\F_q$ such that
\[\sum_{i=1}^{g}\sum_{j\in S_i}\Gl_{ij}g_{ij}(Q)=\sum_{i\in T}\left(\sum_{j\in S_i}\Gl_{ij}g_{ij}\right)(Q) =0.\]
Since $|T|\le \min\{k,g\}=\deg(Q)/m$, it follows from Lemma \ref{lem:3.1} that the function  $\sum_{j\in S_i}\Gl_{ij}g_{ij}=0$ for each $i\in T$. As $\{g_{ij}\}_{j\in S_i}$ are $\F_q$-linearly  independent, we get $\Gl_{ij}=0$ for all $j\in S_i$. This completes the proof.
\end{proof}

By taking $m=r$, we obtain the following result.

\begin{theorem}\label{thm:3.3} If $r\ge \log n$, then there exists an MR $(n,r,h,a)$-LRC  of dimension $k=n-\frac{na}r-h$ over a field of size
\[
\ell\le \left\{\begin{array}{ll}
2^{\min\left\{rk,n \right\}} \le 2^{ n} &\mbox{ if $r\ge \log n$} \\
2^{\min\left\{k\lceil\log n\rceil,\frac nr \lceil\log n\rceil \right\}}&\mbox{ if $r\le \log n$}
\end{array}\right. \]
\end{theorem}
\begin{proof} If $r\ge \log n$, put $m=r$. If $r\le \lceil\log n\rceil$, put $m=\log n$.  Consider the rational function field $\F_2(x)$. To have $g=\frac{n}r$ pairwise coprime polynomials $\{p_i(x)\}_{i=1}^g$ of degree $m$,  it is sufficient to satisfy the inequality $2^m\ge mg=m\times\frac nr$, i.e., $2^r\ge n$ which is the given condition. Now the desired result follows from Lemma \ref{lem:3.2}.
\end{proof}

By considering binary BCH codes, we obtain the following binary codes.
\begin{lemma}\label{lem:3.4} There exists a binary $[r,r-m,\ge d]$-linear code with $m=\lfloor\frac{d-1}2\rfloor\cdot\lceil\log_2r\rceil+1$.
\end{lemma}
\begin{proof} Put $t=\lceil\log_2r\rceil$. Then we have a binary $[2^t, 2^t-1-\lfloor(d-1)/2\rfloor t, d]$-extended BCH code for any $d\ge 2$.

Puncturing $2^t-r$ positions, one gets a binary $[r,r-1-\lfloor(d-1)/2\rfloor t,\ge d]$-linear code.
\end{proof}

Combining the binary BCH codes of Lemma~\ref{lem:3.4} with Lemma~\ref{lem:3.2} applied with rational function field $\F_2(x)$ yields the following theorem.
\begin{theorem}\label{thm:3.5} If $r-a=\Omega(\log n)$, then there exists an MR $(n,r,h,a)$-LRC  of dimension $k=n-\frac{na}r-h$ over a field of size
\[\ell \le 
2r^{\min\left\{k\lfloor\frac{r-a}2\rfloor,\frac {n}r\lfloor\frac{r-a}2\rfloor\right\}}\le 2r^{\frac {n}r\lfloor\frac{r-a}2\rfloor}.\]
\end{theorem}
\begin{proof} Consider the rational function field $\F_2(x)$ and a binary $[r,r-m, r-a+1]$-linear code with $m=\lfloor\frac{r-a}2\rfloor\cdot\lceil\log_2r\rceil+1$.
 To have $g=\frac{n}r$ pairwise coprime polynomials $\{p_i(x)\}_{i=1}^g$ of degree $m$, it is sufficient to satisfy the inequality $2^m\ge mg=m\times\frac nr$. Under the condition that $r-a=\Omega(\log n)$, this inequality is satisfied. Now the desired result follows from Lemma \ref{lem:3.2}.
\end{proof}

\subsection{Constructions via parity-check matrix}
\label{subsec:rational-pc}
To construct parity-check matrices of MR LRCs, we only need to construct matrices $D_i$ given in \eqref{eq:5}. As we will see, the idea of constructing matrices $D_i$  is quite similar to that of constructing matrices $B_i$ in the previous subsection. Our goal is to prove the following theorem.

\begin{theorem}\label{thm:3.7} Let $r,g,a,h,m$ be positive integers with $a\le r$. Suppose that $q \ge r$ is a prime power satisfying $q^m\ge \frac{mn}r$ and  there is a $q$-ary $[r,r-a,a+1]$-linear code. If (i) $m\ge r$;  or (ii) $m<r$ and there exists a $q$-ary $[r,r-m,\ge h+a+1]$-linear code, then there exists an MR $(n,r,h,a)$-LRC with $n=rg$ over a field of size $\ell={q^{\min\{hm,\frac {nm}r\}}}$.
\end{theorem}
\begin{proof}
We can
choose $g\le\left\lceil \frac{q^m}m\right\rceil$ polynomials $p_1(x),p_2(x),\dots,p_g(x)$ of degree $m$ such that $\gcd(p_i(x),p_j(x))=1$ for all $1\le i\neq j\le g$. Then for each $1\le i\le g$, we can form an $\F_q$-vector space \[ V_i:=\left\{\frac{f(x)}{p_i(x)}:\; f(x)\in\F_q[x],\; \deg(f(x))\le m-1\right\} \]
of dimension $m$. Under our assumption about $m$, one can find $r$ functions $g_{i1}(x),\dots,g_{ir}(x)\in V_i$ such that any $h+a$ polynomials out of $\{g_{i1}(x),\dots,g_{ir}(x)\}$ are $\F_q$-linearly independent. 

Choose an irreducible polynomial $Q(x)\in\F_q[x]$ of degree $\min\{hm,\frac {nm}r\}$ and define the matrix
\begin{equation}\label{eq:11}D_i=\left(\begin{array}{ccccccc}
g_{i1}(Q)& g_{i2}(Q)& \cdots & g_{ir}(Q) \\
g_{i1}^q(Q)& g_{i2}^q(Q)&\cdots & g_{ir}^q(Q) \\
\vdots&\vdots&\vdots& \vdots\\
g_{i1}^{q^{h-1}}(Q)&g_{i2}^{q^{h-1}}(Q)& \cdots & g_{ir}^{q^{h-1}}(Q)
\end{array}
\right)\end{equation}
Since $q \ge r \ge a$, we can pick
	$A_i\in\F_q^{a\times r}$ to be a generator matrix of an $[r,a]_q$-MDS code for $1\le i\le g$. Let $D_i$ be the matrix given in \eqref{eq:11}.  Then, we will prove that code $C$ with the matrix $H$ defined in \eqref{eq:5} is an MR $(n,r,h,a)$-LRC over a field of size
\[ \ell=q^{\min\{hm,\frac {nm}r\}} , \]
which will complete the proof of Theorem~\ref{thm:3.7}.

To this end, it is sufficient to prove that the condition (iii) in  Definition \ref{def:2} is satisfied. Let $T_i$ be a subset of $\{(i,1),(i,2),\dots,(i,r)\}$ with $|T_i|=a$ for $1\le i\le g$. Let $S_i$ be a subset of $\{(i,1),(i,2),\dots,(i,r)\}\setminus T_i$  for $i=1,2,\dots,g$ such that $\sum_{i=1}^{g}|S_i|=h$.
Put $A_i=(\ba_{i1},\dots,\ba_{ir})$ and let $\bh_{ij}$ be the $j$th column of the block $i$ in $H$, i.e., $\bh_{ij}=(\bo,\cdots, \ba_{ij}^T,\cdots,\bo,g_{ij}(Q),g_{ij}^q(Q),\dots,$ $g_{ij}^{q^{h-1}}(Q))^T$. To prove the condition (iii) in  Definition \ref{def:2}, it is equivalent to proving that the determinant
$\det((\bh_{ij})_{1\le i\le n, j\in T_i\cup S_i})$ is nonzero for all possible $T_i$ and $S_i$ given above.

Put $M_i=(\ba_{ij})_{j\in T_i}$ and $N_i=(\ba_{ij})_{j\in S_i}$. Denote by $K_i$  and $L_i$ the submatrices $D_i|_{T_i}$ and $D_i|_{S_i}$ of $D_i$ consisting columns indexed by $T_i$ and $S_i$, respectively. Then we have
\[
(\bh_{ij})_{1\le i\le n, j\in T_i\cup S_i}=\left(\begin{array}{c|c|c|c}
M_1,N_1&O&\cdots&O\\ \hline
O&M_2,N_2&\cdots&O \\ \hline
\vdots&\vdots&\ddots&\vdots \\ \hline
O&O&\cdots&M_g,N_g \\ \hline
K_1,L_1&K_2,L_2&\cdots&K_g,L_g
\end{array}
\right)\in\F_{\ell}^{(ag+h)\times (ag+h)}
\]

As $M_i=(\ba_{ij})_{j\in T_i}\in\F_q^{a\times a}$ is invertible, the product
{\small \[
\left(\begin{array}{c|c|c|c}
M_1,N_1&O&\cdots&O\\ \hline
O&M_2,N_2&\cdots&O \\ \hline
\vdots&\vdots&\ddots&\vdots \\ \hline
O&O&\cdots&M_g,N_g \\ \hline
K_1,L_1&K_2,L_2&\cdots&K_g,L_g
\end{array}
\right)\cdot \left(\begin{array}{c|c|c|c}
\begin{array}{cc}I_a,&-M_1^{-1}N_1\\
O&I_{|S_1|}\end{array}
&O&\cdots&O\\ \hline
O&\begin{array}{cc}I_a,&-M_2^{-1}N_2\\
O&I_{|S_2|}\end{array}&\cdots&O \\ \hline
\vdots&\vdots&\ddots&\vdots \\ \hline
O&O&\cdots&\begin{array}{cc}I_a,&-M_g^{-1}N_g\\
O&I_{|S_g|}\end{array}
\end{array}
\right)\]
}
is equal to
\[\left(\begin{array}{c|c|c|c}
M_1,O&O&\cdots&O\\ \hline
O&M_2,O&\cdots&O \\ \hline
\vdots&\vdots&\ddots&\vdots \\ \hline
O&O&\cdots&M_g,O\\ \hline
K_1,-K_1M_1^{-1}N_1+L_1&K_2,-K_2M_2^{-1}N_2+L_2&\cdots&K_g,-K_gM_g^{-1}N_g+L_g
\end{array}
\right)
\]
This implies that $\det((\bh_{ij})_{1\le i\le n, j\in T_i\cup S_i})$ is nonzero if and only if the matrix \begin{equation}\label{eq:12}
(-K_1M_1^{-1}N_1+L_1|-K_2M_2^{-1}N_2+L_2|\cdots|-K_gM_g^{-1}N_g+L_g)\in\F_\ell^{h\times h}\end{equation}
is invertible. Note that the matrix in \eqref{eq:12} is a Moore matrix with the first row:
\begin{equation}\label{eq:13}\left(\left(g_{ij}+\sum_{l\in T_i}\mu_{lj}g_{lj}\right)(Q)\right)_{1\le i\le g,j\in S_i}\end{equation}
for some $\mu_{lj}\in\F_q$. By the property of the Moore determinant, proving  the condition (iii) in  Definition \ref{def:2} is equivalent to showing that the $h$ elements in \eqref{eq:13} are $\F_q$-linearly independent.

Let $R$ be a subset of $\{1,2,\dots,g\}$ such that $S_i\neq\emptyset$ if and only if $i\in R$. Then  $\sum_{i=1}^{g}|S_i|=\sum_{i\in R}|S_i|=h$ and hence $|R|\le \min\{h,g\}=\min\{h,\frac nr\}=\deg(Q)/m$.
Let $\Gl_{ij}\in\F_q$ such that\\
$\sum_{i=1}^{g}\sum_{j\in S_i}\Gl_{ij}\left(g_{ij}+\sum_{l\in T_i}\mu_{lj}g_{lj}\right)(Q)=0$, i.e.,
\[\sum_{i\in R}\sum_{j\in S_i}\Gl_{ij}\left(g_{ij}+\sum_{l\in T_i}\mu_{lj}g_{lj}\right)(Q)=0.\]
By Lemma \ref{lem:3.1}, $\sum_{j\in S_i}\Gl_{ij}g_{ij}+\sum_{l\in T_i}\mu_{lj}\left(\sum_{j\in S_i}\Gl_{ij} \right)g_{lj}=0$ for each $1\le i\le g$. As $\{g_{ij}\}_{j\in T_i\cup S_i}$ are $\F_q$-linearly  independent, we get $\Gl_{ij}=0$ for all $j\in S_i$. This completes the proof.
\end{proof}


We now instantiate Theorem~\ref{thm:3.7} with suitable choices of parameters to deduce the promises parts (i)-(iv) of Theorem~\ref{thm:1.1}.

\subsubsection{The case where $a=1$}
Let $r,h\ge 2$  be  integers. Then there is a $q$-ary $[r,1,r]$-MDS  code for any prime power $q$. Rewriting Theorem \ref{thm:3.7} for  $a=1$ gives the following lemma.
\begin{lemma}\label{lem:3.8} Suppose that  $q^m\ge \frac{mn}r$. If (i) $m\ge r$;  or (ii) $m<r$ and there exists a $q$-ary $[r,r-m,\ge h+2]$-linear code, then there exists an MR $(n,r,h,1)$-LRC   over a field of size $\ell={q^{\min\{hm,\frac {nm}r\}}}$.
\end{lemma}
To apply Lemma \ref{lem:3.8}, we need to find suitable codes and function fields as well. By taking the rational function field $\F_2(x)$ and applying BCH code given in  Lemma~\ref{lem:3.4}, we obtain the following result.
\begin{theorem}\label{thm:3.9} If $r\ge h+2$, then there exists an MR $(n,r,h,1)$-LRC over a field of size
\[ \ell \le \left(\max\left\{\tilde{O}(\frac n{r}), (2r)^{\left\lfloor\frac{h+1}2\right\rfloor}\right\}\right)^{\min\{h,\frac nr\}} \ . \]
\end{theorem}
\begin{proof}  Consider the rational function field $F=\F_2(x)$. Put \[m=\max\left\{\left\lfloor\frac{h+1}2\right\rfloor\cdot\lceil\log_2r\rceil+1,\left\lceil\log_2\left(\frac nr\right)+2\log_2\log_2\left(\frac nr\right)\right\rceil\right\}.\]
 Then $\frac nr\le \frac1m2^m$. This implies that there are $\frac nr$ places of degree $m$ in $\F_2(x)$. By Lemma \ref{lem:3.4}, there exists a binary $[r,r-m,\ge h+2]$-linear code. It follows from
 Lemma \ref{lem:3.8} that there exists an MR $(n,r,h,1)$-LRC  over a field of size $2^{\min\{mh,m\frac nr\}}$. By choice of our parameters, the desired result follows.
\end{proof}

\begin{theorem}\label{thm:3.10} There exists  an MR $(n,r,h,1)$-LRC  over a field of size
	\[\ell \le \left( \max\left\{ \tilde{O}(\frac nr), 2^r\right\}\right)^{\min\{h,\frac nr\}} \ . \]
\end{theorem}
\begin{proof} Consider the rational function field $\F_2(x)$. Put $m=\max\{r,\left\lceil\log_2\left(\frac nr\right)+2\log_2\log_2\left(\frac nr\right)\right\rceil\}$. Then $\frac nr\le \frac1m2^m.$  The desired follows from Lemma \ref{lem:3.8}.
\end{proof}

\begin{rmk}\label{rmk:4}{\rm  Theorem \ref{thm:3.10} gives a better bound on the field size than Theorem \ref{thm:3.9} for $h>\frac{2r}{\log_2r}-1$, while Theorem \ref{thm:3.9} gives a better bound on the field size than Theorem \ref{thm:3.10} for $h<\frac{2r}{\log_2r}-1$.
}\end{rmk}

\subsubsection{The case where $2\le a\le r-1$}
\begin{lemma}\label{lem:3.11} Let $a\le r\le q+1$ and $m\ge h+a$. If  $q^m\ge \frac{mn}r$, then there exists an MR $(n,r,h,a)$-LRC code  over a field of size $\ell={q^{\min\{mh,\frac{mn}r\}}}$.
\end{lemma}
\begin{proof} Under the assumption that  $a\le r\le q+1$ and $m\ge h+a$, we have an $[r,r-a,a+1]_q$-MDS code and an $[r,r-m,h+a+1]_q$-linear code. The desired result follows from Theorem \ref{thm:3.7}.
\end{proof}

\begin{theorem}\label{thm:3.12} There exists an MR $(n,r,h,a)$-LRC  over a field of size
\[\ell \le \left(\max\left\{\tilde{O}(\frac nr), (2r)^{h+a}\right\}\right)^{\min\{h,\frac nr\}} \ . \]
\end{theorem}
\begin{proof} Let $q$ be the smallest prime power such that $q-1\ge r$. We may take $q$ to be a power of two, so that $q \le 2r$.  Consider the rational function field $F=\F_q(x)$ and let \[m=\max\left\{h+a, \left\lceil\log_q\left(\frac nr\right)+2\log_q\log_q\left(\frac nr\right)\right\rceil\right\}.\]
 Then $\frac nr\le \frac1m q^m$. The desired result follows from Theorem \ref{thm:3.7}.
\end{proof}

\begin{rmk}\label{rmk:3}{\rm  The field size $\ell \le \tilde{O}\left(\max\left\{\frac nr, r^{h+a}\right\}^h\right)$ in Theorem \ref{thm:3.12} was already given in \cite[Corollary 11]{GYBS17}. Here we provide better result for $h>\frac nr$ via a different approach.
}\end{rmk}

\begin{theorem}\label{thm:3.13} There exists  an MR $(n,r,h,a)$-LRC  over a field of size
	\[ \ell \le \left(\max\left\{\tilde{O}(\frac nr), (2r)^r\right\}\right)^{\min\{h,\frac nr\}} \ . \]
\end{theorem}
\begin{proof} Put $q=2^{\lceil \log_2r\rceil}$. Then $2r \ge q\ge r$ and hence we have a $q$-ary $[r,a]$-MDS code for any $a\le r$. Put $m=\max\{r,\left\lceil\log_q\left(\frac nr\right)+2\log_q\log_q\left(\frac nr\right)\right\rceil\}$. Then $\frac nr\le \frac1mq^m.$  The desired follows from Theorem \ref{thm:3.9}.
\end{proof}

\begin{rmk}\label{rmk:5}{\rm  Theorem \ref{thm:3.13} gives a better bound on the field size than Theorem \ref{thm:3.12} for $h+a>r$, while Theorem \ref{thm:3.12} gives a better bound on the field size than Theorem \ref{thm:3.13} for $h+a<r$.
}\end{rmk}

\section{Explicit construction via general function fields}

The construction via rational function fields given in Section 3 can be easily generalized to arbitrary function fields. We begin with some preliminaries on function fields.

\subsection{Background on function fields}
A function field over $\F_q$ is a field $F$ containing $\F_q$ satisfying that there is a transcendental element $x\in F$  over $\F_q$ such that $F$ is an algebraic extension over $\F_q(x)$. If $\F_q$ is algebraic closed in $F$, then $\F_q$ is called the full constant field of $F$, denoted by $F/\F_q$.

Each discrete valuation $\nu$ from $F/\F_q$ to $\ZZ\cup\{+
\infty\}$ defines a local ring $O=\{f\in F:\; \nu(f)\ge 0\}$. The maximal ideal $P$ of $O$ is called a {\it place}. We denote  the valuation $\nu$ and the local ring $O$ corresponding to $P$ by $\nu_P$ and $O_P$, respectively. The residue class field $O_P/P$, denoted by $F_P$, is a finite extension of $\F_q$. The extension degree $[F_P:\F_q]$ is called {\it degree} of $P$, denoted by $\deg(P)$. A place of degree one is called a {\it rational} place. For two functions $f,g\in F$ and a place $P$, we have $\nu_P(f+g)\ge \min\{\nu_P(f),\nu_P(g)\}$ and the equality holds if $\nu_p(f)\neq\nu_P(g)$ (note that we set $\nu_P(0)=+\infty$). In particular, this implies that $f+g\neq 0$ if $\nu_P(f)\neq\nu_P(g)$.
\footnote{Geometrically, a place corresponds to a point on an algebraic curve, and the valuation of a function $f$ at a point $P$ is the order of vanishing of $f$ at $P$. (If $f$ blows up at $P$, i.e., has a pole at $P$, then the valuation is negative, and equal to the zero order of $1/f$ at $P$.) When a function $f$ with non-negative valuation at $P$ is evaluated at $P$ we get a value in the residue field $O_P/P$ --- one can think of the coordinates of the point $P$ as belong to the extension field $O_P/P$ of $\F_q$.}

Let $\PP_F$ denote the set of places of $F$ and let $\PP_F(m)$ denote the set of places of degree $m$ of $F$. A divisor $D$ of $F$ is a formal sum $\sum_{P\in\PP_F}m_PP$, where $m_P\in\ZZ$ are equal to $0$ except for finitely many $P$. The degree of $D$ is defined to be $\deg(D)=\sum_{P\in\PP_F}m_P\deg(P)$. We say that $D$ is positive, denoted by $D\ge 0$, if $m_P\ge 0$ for all $P\in\PP_F$. For a nonzero function $f$, the principal divisor $(f)$ is defined to be $\sum_{P\in\PP_F}\nu_P(f)P$. Then the degree of the principal divisor $(f)$ is $0$.  The Riemann-Roch space associated with a divisor $D$, denoted by $\mL(D)$, is defined by
\begin{equation}\label{eq:6}
\mL(D):=\{f\in F\setminus\{0\}:\; (f)+D\ge 0\}\cup\{0\}.
\end{equation}
Then $\mL(D)$ is a finite dimensional space over $\F_q$. By the Riemann-Roch theorem \cite{St93}, the dimension of $\mL(D)$, denoted by $\dim_{\F_q}(D)$, is lower bounded by $\deg(D)-\g+1$, i.e., $\dim_{\F_q}(D)\ge \deg(D)-\g+1$, where $\g$ is the genus of $F$. Furthermore, $\dim_{\F_q}(D)= \deg(D)-\g+1$ if $\deg(D)\ge 2\g-1$. In addition, we have the following results \cite[Lemma 1.4.8 and Corollary 1.4.12(b)]{St93}:
\begin{itemize}
	\item[(i)]
	If $\deg(D)<0$, then $\dim_{\F_q}(D)= 0$;
	\item[(ii)] For a positive divisor $G$, we have $\dim_{\F_q}(D)-\dim_{\F_q}(D-G)\le \deg(G)$, i.e., $\dim_{\F_q}(D-G)\ge \dim_{\F_q}(D)-\deg(G)$.
\end{itemize}

For a nonzero function $f$, we denote by $(f)_0$ and $(f)_{\infty}$ the zero and pole divisors of $f$, respectively, i.e.,
\[(f)_0=\sum_{P\in\PP_F, \nu_P(f)>0}\nu_P(f)P \quad
and \quad (f)_\infty=-\sum_{P\in\PP_F, \nu_P(f)<0}\nu_P(f)P.\]
We have $\deg((f)_0)=\deg((f)_\infty)$ since the degree of the principal divisor $(f) = (f)_0 - (f)_\infty$ equals $0$.

Let $\ell=q^m$ for a positive integer and let $F/\F_q$ be a function field. Then every place of degree $m$ of $F$ splits into $m$ $\F_\ell$-rational place in the constant field extension $\F_\ell\cdot F$. We also call an $\F_\ell$-rational place of $\F_\ell\cdot F$ an $\F_\ell$-rational place of $F$. Let $N_\ell$ denote the number of $\F_\ell$-rational places of $F$ and let $B_i$ denote the number of places of $F$ of degree $i$. Then we have the relation (see \cite[page 178]{St93})
\begin{equation}\label{eq:7}
N_m=\sum_{d|m}d\cdot B_d.
\end{equation}
It immediately follows from \eqref{eq:7}  that $\sum_{d|m}B_d\ge \lceil \frac{N_m}{m}\rceil$. For each divisor   $d$ of $m$ and each place $P$ of degree $d$ of $F$,  $(m/d)P$ is a positive divisor of degree $m$. Thus, we have at least
\begin{equation}\label{eq:8}g=\sum_{d|\ell}B_d\ge \left\lceil \frac{N_\ell}{\ell}\right\rceil\end{equation}
positive divisors of degree $m$ whose supports are pairwise disjoint.

\subsubsection{Hermitian function field}\label{subsec:2.4}
Let $q$ be a prime power and let $s=q^m$ be a even power of a prime. The Hermitian function field $F/\F_{q}$ is given by $F=\F_q(x,y)$, where $x,y$ are two transcendental elements over $\F_q$ satisfying the equation
\[y^{\sqrt{s}}+y=x^{\sqrt{s}+1}.\]
The genus of $F$ is $\g(F)=\frac12{\sqrt{s}(\sqrt{s}-1)}\le\frac 12 s$. The number of $\F_s$-rational places of $F$ is $1+s^{3/2}$. One of these is the ``point at infinity" which is the unique common pole  of $x$ and $y$. The other $s^{3/2}$ places come from the $\F_s$-rational places lying over the unique zero $P_\Ga$ of $x-\Ga$ for each $\Ga\in\F_s$. Note that for every $\Ga\in\F_s$, $P_\Ga$ splits completely in $\F_s\cdot F$, i.e., there are $\sqrt{s}$ $\F_s$-rational places lying over $P_\Ga$.
Intuitively, one can think of the $\F_s$-rational places of $F$ (besides $\Pin$) as being given by pairs $(\Ga,\Gb)\in \F_s^2$ that satisfy $\Gb^{\sqrt{s}}+\Gb=\Ga^{\sqrt{s}+1}$. For each value of $\Ga \in \F_s$, there are precisely $\sqrt{s}$ solutions to $\beta \in \F_s$ satisfying $\Gb^{\sqrt{s}}+\Gb=\Ga^{\sqrt{s}+1}$.

Thus, the genus $\g(F)$ of $F$ satisfies $2\g(F)\le N_m^{2/3}$.

\subsubsection{Garcia-Stichtenoth tower}\label{subsec:2.5}
The Garcia-Stichtenoth
tower is an optimal one in the sense that the ratio of number of
rational places against genus achieves the maximal possible value.
Again let $q$ be a prime power and let $s=q^m$ be a even power of a prime. The Garcia-Stichtenoth towers that
we are going to use for our code construction were discussed in
\cite{GS95,GS96}. The reader may refer to \cite{GS95,GS96} for the
detailed background on the Garcia-Stichtenoth function field tower.  There
are two optimal Garcia-Stichtenoth towers that are equivalent. For
simplicity, we introduce the tower defined by the following recursive
equations \cite{GS96}
\begin{equation}\label{eq:9}
x_{i+1}^{\sqrt{s}}+x_{i+1}=\frac{x_i^{\sqrt{s}}}{x_i^{\sqrt{s}-1}+1},\quad i=1,2,\dots,t-1.
\end{equation}
Put  $K_t=\F_q(x_1,x_2,\dots,x_{t})$ for $t\ge 2$.

\medskip
\noindent {\bf Rational places.}
The function field $K_t$ has at least $s^{(t-1)/2}(s-\sqrt{s})+1$ $\F_s$-rational places. One of these is the ``point at infinity" which is the unique pole $\Pin$ of $x_1$ (and is fully ramified). The other $s^{(t-1)/2}(s-\sqrt{s})$ come from the $\F_s$-rational places lying over the unique zero of $x_1-\Ga$ for each $\Ga\in\F_s$ with $\Ga^{\sqrt{s}}+\Ga\not=0$. For every $\Ga\in\F_s$ with $\Ga^{\sqrt{s}}+\Ga\not=0$, the unique zero of $x_1-\Ga$ splits completely in $\F_s\cdot K_t$, i.e., there are $s^{(t-1)/2}$ $\F_s$-rational places lying over the zero of $x_1-\Ga$. Let $\PP$ be the set of all the rational places lying over the zero of $x_1-\Ga$ for all $\Ga\in\F_s$ with $\Ga^{\sqrt{s}}+\Ga\not=0$. Then, intuitively, one can think of the $s^{(t-1)/2}(s-\sqrt{s})$ $\F_s$-rational places in $\PP$ as being given by $t$-tuples $(\Ga_1,\Ga_2,\dots,\Ga_t)\in \F_s^t$ that satisfy $\Ga_{i+1}^{\sqrt{s}}+\Ga_{i+1}=\frac{\Ga_i^{\sqrt{s}}}{\Ga_i^{\sqrt{s}-1}+1}$ for $i=1,2,\dots,t-1$ and $\Ga_1^{\sqrt{s}}+\Ga_1\not=0$. For each value of $\Ga \in \F_s$, there are precisely $\sqrt{s}$ solutions to $\beta \in \F_s$ satisfying $\Gb^{\sqrt{s}}+\Gb=\frac{\Ga^{\sqrt{s}}}{\Ga^{\sqrt{s}-1}+1}$, so the number of such $t$-tuples is $s^{(t-1)/2}(s-\sqrt{s})$ (there are $s-\sqrt{s}$ choices for $\Ga_1$, and then $\sqrt{s}$ choices for each successive $\Ga_i$, $2 \le i  \le t$).

\medskip
\noindent {\bf Genus.}
The genus $\g_t$ of the function field $K_t$ is given by
\[\g_t=\left\{\begin{array}{ll}
((\sqrt{s})^{t/2}-1)^2&\mbox{if $t$ is even}\\
((\sqrt{s})^{(t-1)/2}-1)((\sqrt{s})^{(t+1)/2}-1)&\mbox{if $t$ is odd.}\end{array}
\right.\]
Thus  the genus of $\g(K_t)$ is upper bounded by $s^{t/2}$ and the number $N_m$ of $\F_s$-rational points is lower bounded by
\[N_m\ge \left(q^{\frac m2}-1\right)\g(K_t).\]
By \eqref{eq:8}, there are at least $\frac{\left(q^{\frac m2}-1\right)\g(K_t)}{m}$ positive divisors of degree $m$ whose supports are pairwise disjoint.

\subsection{MR LRCs from function fields}
We only generalize the constructions of MR LRCs via parity-check matrices given in Section~\ref{subsec:rational-pc}.

Let $q$ be a  prime power and let $a, r,h, g$  be  integers with $a\le r\le q+1$. Let $F/\F_q$ be a function field of genus $\g$. Let $P_1,P_2,\dots,P_g$ be $g$ positive divisors of   degree $r$ whose supports are pairwise disjoint. Let $G$ be a divisor of degree $2\g-1$.
By Riemann-Roch, $\dim \mathcal{L}(G) = \g$. Assume that $\{ f_1,f_2,\dots,f_{\g}\}$ is a basis of $\mL(G)$. For each $i$, extend this basis to a basis $\{f_1,f_2,\dots,f_{\g},f_{i1},f_{i2},\dots,f_{ir}\}$ of $\mL(G+P_i)$.

Let $Q$ be a place of degree $2\g+\min\{hr,n\}$ and define the matrix
\begin{equation}\label{eq:14}D_i=\left(\begin{array}{ccccccc}
f_{i1}(Q)& f_{i2}(Q)& \cdots & f_{ir}(Q) \\
f_{i1}^q(Q)& f_{i2}^q(Q)&\cdots & f_{ir}^q(Q) \\
\vdots&\vdots&\vdots& \vdots\\
f_{i1}^{q^{h-1}}(Q)&f_{i2}^{q^{h-1}}(Q)& \cdots & f_{ir}^{q^{h-1}}(Q)
\end{array}
\right)
\end{equation}
By mimicking the proof of Theorem \ref{thm:3.7}, we have the following result.
\begin{lemma}\label{lem:4.1} Let $A_i\in\F_q^{a\times r}$ be a generator matrix of an $[r,a]_q$-MDS code for $1\le i\le g$. Let $D_i$ be the matrix given in \eqref{eq:14}. Put $\ell=q^{2\g+\min\{hr,n\}}$. Then the $\ell$-ary code $C$ with the matrix $H$ defined in \eqref{eq:5} is an MR $(n,r,h,a)_\ell$-LRC.
\end{lemma}
\begin{proof}
By mimicking the proof of Theorem \ref{thm:3.7}, it is
sufficient to show that
\[\left(\left(f_{ij}+\sum_{l\in T_i}\mu_{lj}f_{lj}\right)(Q)\right)_{1\le i\le f,j\in S_i}\]
is $\F_q$-linearly independent
for any $\mu_{lj}\in\F_q$.

Let $T$ be a subset of $\{1,2,\dots,g\}$ such that $S_i\neq\emptyset$ if and only if $i\in T$. Then  $\sum_{i=1}^{g}|S_i|=\sum_{i\in T}|S_i|=h$ and hence $|T|\le \min\{h,g\}=\min\{h,\frac nr\}$.
Let $\Gl_{ij}\in\F_q$ such that\\
$\sum_{i=1}^{g}\sum_{j\in S_i}\Gl_{ij}\left(f_{ij}+\sum_{l\in T_i}\mu_{lj}f_{lj}\right)(Q)=0$, i.e.,
\[\sum_{i\in T}\sum_{j\in S_i}\Gl_{ij}\left(f_{ij}+\sum_{l\in T_i}\mu_{lj}f_{lj}\right)(Q)=0.\]
The function $\sum_{i\in T}\sum_{j\in S_i}\Gl_{ij}\left(f_{ij}+\sum_{l\in T_i}\mu_{lj}f_{lj}\right)$ belongs to the Riemann-Roch space $\mL(G+\sum_{i\in T}P_i-Q)$. As $\deg(G+\sum_{i\in T}P_i-Q)\le 2\g-1+r|T|-\deg(Q)<0$, we  have\\   $\sum_{i\in T}\sum_{j\in S_i}\Gl_{ij}\left(f_{ij}+\sum_{l\in T_i}\mu_{lj}f_{lj}\right)=0$.
Rewrite this equality into
\[\sum_{i\in T}\left(\sum_{j\in S_i}\Gl_{ij}f_{ij}+\sum_{l\in T_i}\mu_{lj}\left(\sum_{j\in S_i}\Gl_{ij} \right)f_{lj}\right)=0.   \]
This forces that $\sum_{j\in S_i}\Gl_{ij}f_{ij}+\sum_{l\in T_i}\mu_{lj}\left(\sum_{j\in S_i}\Gl_{ij} \right)f_{lj}=0$ for each $1\le i\le g$. As $\{f_{ij}\}_{j\in T_i\cup S_i}$ are $\F_q$-linearly  independent, we get $\Gl_{ij}=0$ for all $j\in S_i$. This completes the proof.
\end{proof}
Consequently, we have the following theorem.
\begin{theorem}\label{thm:4.2} Let $r,g,a,h$ be positive integers with $a\le r\le q+1$. If there is a function field $F/\F_q$ of genus $\g$ with $g$ positive divisors of degree $m$ whose supports are disjoint, then there exists an MR $(n,r,h,a)$-LRC with $n=rg$ over a field of size $\ell=q^{2\g+\min\{hr,n\}}$.
\end{theorem}

Finally, let us instantiate the above result with the Hermitian function fields and the Garcia-Stichtenoth tower, to deduce Parts (v) and (vi) promised in Theorem~\ref{thm:1.1} respectively. Note that both the results below kick-in for block lengths which are asymptotically at least $r^{O(r)}$, which is why we have the condition $r \le O(\frac{\log n}{\log \log n})$ in the statement of  Theorem~\ref{thm:1.1}, Parts (v), (vi).
\begin{theorem}\label{thm:4.3}
Let $a \le r$ be integers. Then there are infinitely many $n \ge r^{\Omega(r)}$ such that there is MR $(n,r,h,a)$-LRC over a field of size at most $n^{\frac{2h}{3}\left(1+{\Ge}\right)}$ for any desired $\Ge \in (0,0.5)$ provided $hr \ge \Omega\left( \frac{ n^{\frac23}}{\Ge}\right)$.
\end{theorem}
\begin{proof}
Let $m$ be a parameter that is at least $\lceil \frac{\log r}{2} \rceil$ and let $q=4^m$.
Then $q \ge r$ and we have a $q$-ary $[r,a]$-MDS code for any $a\le r$.
	Consider the Hermitian function field $F/\F_q$ defined by the equation  $y^{q^{0.5r}}+y=x^{q^{0.5r+1}}$. By Subsection \ref{subsec:2.4}, the genus $\g$ of $F$ satisfies $2\g\le N_r^{2/3}$, where $N_r=1+q^{3r/2}$ is the number of $\F_{q^r}$-rational places of $F$.  By \eqref{eq:8}, there are at least $\lceil \frac{N_r}{r}\rceil$ positive divisors of degree $r$ whose supports are pairwise disjoint. Let $n=r\times \lceil \frac{N_r}{r} \rceil$.
Note that $N_r  \le n < N_r + r$, so $q^{3r/2} < n \le q^{3r/2} + r$. Note that the smallest value of $n$ is $r^{O(r)}$ since we need $q \ge r$, and as we increase $m$ we get a family of codes with larger block lengths.

	  By  Theorem \ref{thm:3.9}, there exists an MR $(n,r,h,a)$-LRC over a field of size
	\[q^{2\g+hr} = q^{hr\left(1+\frac{q^r}{hr}\right)} \le n^{2h/3 \left(1+\frac{n^{2/3}}{hr}\right)} \le  n^{h\left(1+{\Ge}\right)\times\frac23} \ \]
	where the last inequality follows from our assumed lower bound on $h$. .	
   \end{proof}

We now turn to a similar result using the Garcia-Stichtenoth tower of function fields.

\begin{theorem}\label{thm:4.4} Let $a \le r$ be positive integers and let $\Ge\in(0,0.5)$.  Then there are infinitely many $n \ge r^{\Omega(r/\Ge)}$ such that there is MR $(n,r,h,a)$-LRC over a field of size at most $n^{\Ge h}$ provided $hr \ge \Omega\left(n^{1-\Ge}\right)$.
\end{theorem}
\begin{proof}
Let $m$ be a parameter that is at least $\lceil \frac{\log r}{2} \rceil$ and let $q=4^m$.
Then $q \ge r$ and we have a $q$-ary $[r,a]$-MDS code for any $a\le r$. Put $s=q^r$ and $t= \lceil \frac{4}{\Ge} \rceil$.
Consider the $t$'th function field $K_t=\F_q(x_1,x_2,\dots,x_{t})$ in the Garcia-Stichtenoth tower defined
 by
 \begin{equation}\label{eq:15}
x_{i+1}^{\sqrt{s}}+x_{i+1}=\frac{x_i^{\sqrt{s}}}{x_i^{\sqrt{s}-1}+1},\quad i=1,2,\dots,t-1.
\end{equation}
 in Subsection \ref{subsec:2.5}.

Then the number $N_r$ of $\F_{q^r}$-rational places of $K_t$ is
$s^{(t-1)/2}(s-\sqrt{s})+1=q^{rt/2}(q^{r/2}-1)+1$.
Thus,  there are at least   $g=\lceil \frac 1r N_r \rceil$ positive divisors of degree $r$ of $K_t$  whose supports are pairwise disjoint. Let $n=gr$ be the block length of the $r$-local LRC that we will construct. Note that $q^{rt/2}(q^{r/2}-1) < n  \le q^{r(t+1)/2}$.  Note that the smallest value of $n$ is $r^{O(rt)}=r^{O(r/\Ge)}$ since we need $q \ge r$, and as we increase $m$ we get a family of codes with larger block lengths.

By Theorem \ref{thm:3.9}, there exists an MR $(n,r,h,a)$-LRC  over a field of size
\begin{equation}
\label{eq:GS-field-size}
q^{2\g+hr}=q^{hr\left(1+\frac{2\g}{hr}\right)} \ .
\end{equation}
We have $q^r \le n^{2/t} \le n^{\Ge/2}$. Also
\[ \g=\g(K_t)\le N_r/(q^{r/2}-1) \le 2nq^{-r/2} \le 2n^{1-2/(t+1)}\le 2 n^{1-\Ge} \  \]
Therefore, the field size in \eqref{eq:GS-field-size} is at most $n^{\Ge h}$ assuming $hr \ge 4 n^{1-\Ge} \ge 2\g$.
\end{proof}

\end{document}